\documentclass[10pt,twocolumn]{article}
\usepackage[letterpaper,margin=0.75in,columnsep=0.3in]{geometry}
\usepackage[T1]{fontenc}
\usepackage{xspace}
\usepackage{microtype}
\usepackage{amsmath,amsthm}
\usepackage{newtxtext,newtxmath}
\usepackage{graphicx}
\usepackage{booktabs}
\usepackage{multirow}
\usepackage{xcolor}
\usepackage{tikz}
\usetikzlibrary{arrows.meta,positioning,fit,calc,backgrounds,shapes.geometric}
\usepackage[numbers,sort&compress]{natbib}
\usepackage[hidelinks]{hyperref}
\usepackage[capitalise,nameinlink]{cleveref}
\usepackage{algorithm}
\usepackage[noend]{algpseudocode}
\usepackage{enumitem}
\setlist{nosep,leftmargin=*}

\newtheorem{proposition}{Proposition}

\newcommand{\ms}{\,ms\xspace}
\newcommand{\s}{\,s\xspace}
\newcommand{\GBs}{\,GB/s\xspace}

\newcommand{\NAOnehitFiveZero}{0.89}

\newcommand{\NAOnenhit}{57}
\newcommand{\NAOnepFiveZero}{0.38}

\newcommand{\NAOnepNineNine}{1.34}

\newcommand{\NBOnehitFiveZero}{0.74}

\newcommand{\NBOnenhit}{4}
\newcommand{\NBOnepFiveZero}{0.31}

\newcommand{\NBOnepNineNine}{1.16}

\newcommand{\NBTwohitFiveZero}{0.29}

\newcommand{\NBTwonhit}{60}
\newcommand{\NBTwopFiveZero}{0.32}

\newcommand{\NBTwopNineNine}{0.36}

\newcommand{\NBThreehitFiveZero}{0.88}

\newcommand{\NBThreenhit}{62}
\newcommand{\NBThreepFiveZero}{0.32}

\newcommand{\NBThreepNineNine}{1.11}

\newcommand{\NBFivehitFiveZero}{0.83}

\newcommand{\NBFivenhit}{62}
\newcommand{\NBFivepFiveZero}{0.32}

\newcommand{\NBFivepNineNine}{1.00}

\newcommand{\NBSixhitFiveZero}{0.83}

\newcommand{\NBSixnhit}{60}
\newcommand{\NBSixpFiveZero}{0.32}

\newcommand{\NBSixpNineNine}{0.95}
\newcommand{\NBstar}{2.53}

\newcommand{\NHOnebcached}{98.5}

\newcommand{\NHOnebmax}{86.72}
\newcommand{\NHOnebmaxrun}{59}

\newcommand{\NHOnebpFiveZero}{9.70}
\newcommand{\NHOnebpNineZero}{61.17}
\newcommand{\NHOnebpNineNine}{83.32}
\newcommand{\NHOnebthirdOne}{3.52}
\newcommand{\NHOnebthirdTwo}{9.02}
\newcommand{\NHOnebthirdThree}{52.30}
\newcommand{\NHTwobcached}{50.6}

\newcommand{\NHTwobmax}{4.04}

\newcommand{\NHTwobpFiveZero}{3.20}
\newcommand{\NHTwobpNineZero}{3.70}
\newcommand{\NHTwobpNineNine}{3.82}
\newcommand{\NHTwobthirdOne}{1.20}
\newcommand{\NHTwobthirdTwo}{3.24}
\newcommand{\NHTwobthirdThree}{3.56}
\newcommand{\NHThreecached}{98.5}

\newcommand{\NHThreemax}{1.25}

\newcommand{\NHThreepFiveZero}{0.41}
\newcommand{\NHThreepNineZero}{0.50}
\newcommand{\NHThreepNineNine}{0.80}
\newcommand{\NHThreethirdOne}{0.41}
\newcommand{\NHThreethirdTwo}{0.41}
\newcommand{\NHThreethirdThree}{0.41}
\newcommand{\NHFourcached}{98.5}

\newcommand{\NHFourmax}{1.07}

\newcommand{\NHFourpFiveZero}{0.41}
\newcommand{\NHFourpNineZero}{0.48}
\newcommand{\NHFourpNineNine}{0.73}
\newcommand{\NHFourthirdOne}{0.42}
\newcommand{\NHFourthirdTwo}{0.41}
\newcommand{\NHFourthirdThree}{0.41}
\newcommand{\NHFivebcached}{98.5}

\newcommand{\NHFivebmax}{0.86}

\newcommand{\NHFivebpFiveZero}{0.40}
\newcommand{\NHFivebpNineZero}{0.46}
\newcommand{\NHFivebpNineNine}{0.67}
\newcommand{\NHFivebthirdOne}{0.41}
\newcommand{\NHFivebthirdTwo}{0.40}
\newcommand{\NHFivebthirdThree}{0.40}
\newcommand{\NNbreak}{17}
\newcommand{\NPRODsubmax}{4.86}
\newcommand{\NPRODsubpFiveZero}{0.35}
\newcommand{\NPRODsubpNineZero}{0.41}
\newcommand{\NPRODsubpNineNine}{1.14}
\newcommand{\NRECsubmax}{2.76}
\newcommand{\NRECsubpFiveZero}{0.31}
\newcommand{\NRECsubpNineZero}{0.32}
\newcommand{\NRECsubpNineNine}{1.28}
\newcommand{\NRrate}{52.2}
\newcommand{\NSUBtokmax}{17,889}
\newcommand{\NSUBtokmin}{14,715}
\newcommand{\NSUBtokpFiveZero}{16,390}
\newcommand{\NSUBtoktotal}{7.35}
\newcommand{\NTloadeight}{0.074}

\newcommand{\NTzero}{0.33}
\newcommand{\Nlam}{1.21}
\newcommand{\Npredone}{0.79}
\newcommand{\Nrhoeight}{0.09}

\newcommand{\Nrhosust}{1.13}
\newcommand{\NsKB}{48.4}

\title{\Large\bfseries When the KV Cache Is Slower Than the GPU:\\
Scaling Remote KV Tiers for Context-Parallel Prefill}
\author{Muhammad Ibrahim\\
Multiverse Computing\\
\texttt{muhammad.ibrahim@multiversecomputing.com}}
\date{}

\begin{document}
\maketitle

\begin{abstract}
Reusing the attention KV cache across requests is a standard way to cut prefill
cost in long-context LLM serving, and production systems increasingly back the
per-engine cache with a remote tier shared by every engine. We study such a tier
in a disaggregated deployment of GLM-5.3, an FP8 mixture-of-experts model with
multi-head latent attention (MLA), on 8$\times$B300 nodes running prefill context
parallelism (PCP8). Enabling the remote tier (LMCache with a Valkey store) left
median time-to-first-token unchanged but raised the 16k-token p99 from 3.96\s to
17.3\s. Replaying a production request stream on an isolated engine, we find that
a remote hit took \NBThreehitFiveZero\s against \NBOnepFiveZero\s to recompute the
prompt, and that synchronous loading turned a hit-heavy stream into an unstable
queue: p99 \NHOnebpNineNine\s with the GPUs 7\% busy. A two-term cost model
explains both: a hit pays off only if the tier delivers more than $B^*=sR$
(\NBstar\GBs for this model and GPU), and a fixed per-hit engine cost
($T_0\approx\NTzero$\s) means that hits on prompts shorter than about
\NNbreak k tokens cannot beat recompute at any bandwidth. Guided by the model,
we (i) shard the MLA KV store across PCP ranks, with a proof that the existing
minimum-over-ranks lookup returns the longest cached prefix and a load protocol
that can never desynchronize collectives; (ii) partition the remote tier by
rank, one server instance per rank with no coordination, which raises GET
bandwidth $6.3\times$ to 10.76\GBs; (iii) add asynchronous loading to the
sharded store; and (iv) fix three client bugs that caused unbounded stalls. On
the hit-heavy stream, with 98.5\% of prompt tokens served from the remote tier,
p99 falls from \NHOnebpNineNine\s to \NHFivebpNineNine\s. The sharded store
raises prefill throughput by 14\%, and by 34\% together with an upstream kernel
fix. All changes are submitted upstream to LMCache.
\end{abstract}

\section{Introduction}
\label{sec:intro}

For long prompts, prefill dominates time-to-first-token (TTFT): reading a
200k-token prompt takes seconds even on the newest accelerators. Serving systems
therefore keep the attention key/value (KV) cache of earlier prompts and reuse it
when a new prompt shares a prefix~\cite{vllm,sglang,cachedattention,mooncake,lmcache}.
A common layout has two tiers: a fast tier in each engine's host memory, and a
larger remote tier that every engine can read and that survives restarts.
The remote tier is attractive because it decouples cache capacity from engine
count and lets one engine reuse what another computed.

The implicit assumption behind a remote tier is that \emph{loading} KV is cheaper
than \emph{recomputing} it. This paper reports a production setting in which the
assumption failed, explains why with a simple model, and shows how to restore it.

\paragraph{Setting.} We serve GLM-5.3, a large mixture-of-experts (MoE) model with
DeepSeek-style multi-head latent attention (MLA)~\cite{deepseekv2,deepseekv3},
in FP8 on nodes with eight NVIDIA B300 GPUs. Prefill and decode run on separate
nodes~\cite{distserve,splitwise}. Each prefill node runs one vLLM~\cite{vllm}
engine with \emph{prefill context parallelism} across its eight GPUs (PCP8): the
prompt is split along the sequence, the ranks exchange KV with collectives, and
all eight ranks advance in lockstep. LMCache~\cite{lmcache} stores KV in 256-token
chunks in host memory (L1) and in a Valkey~\cite{valkey} remote store (L2).

\paragraph{Problem.} Turning on L2 for an enterprise SLA workload of 9k--200k-token
prompts left median TTFT unchanged but inflated the tail: p99 rose from 3.96\s to
17.3\s at 16k tokens and from 6.06\s to 25.0\s at 32k. The deployment was rolled
back to L1 only.

\paragraph{Findings.} We rebuilt one production engine's request stream from logs,
replayed it against an isolated engine with the identical image and
configuration, and profiled it (\S\ref{sec:method}). Three effects explain the
regression:
\begin{enumerate}
\item \textbf{A hit was slower than a miss.} A 16k-token prompt took
  \NBOnepFiveZero\s to recompute but \NBThreehitFiveZero\s to load from L2
  (\S\ref{sec:eval-cost}). On a B300 node the GPU prefills about
  \NRrate k tokens/s; the L2 path delivered well under the bandwidth needed to
  keep up with that.
\item \textbf{Synchronous loading is unstable under hit-heavy load.} A load
  inside the engine step stalls all eight PCP ranks. When the arrival rate times
  the per-hit service time exceeds one, the engine becomes an overloaded queue:
  median latency grew from \NHOnebthirdOne\s to \NHOnebthirdThree\s within six
  minutes while the GPUs were 7\% busy (\S\ref{sec:eval-hit}).
\item \textbf{Rare client bugs caused unbounded stalls.} A lost-completion race
  and a missing-key read that never returns each froze an engine for 3--10\s per
  request, or permanently (\S\ref{sec:bugs}).
\end{enumerate}

\paragraph{Contributions.}
\begin{itemize}
\item A \textbf{cost model} for remote KV hits with two parameters,
  break-even bandwidth $B^*=sR$ and a fixed per-hit cost $T_0$, validated on
  configurations that were not used to fit it (\S\ref{sec:model}).
\item A \textbf{sharded MLA KV store for context parallelism} that spreads
  store, lookup and load across ranks, with a proof that the scheduler's
  existing minimum-over-ranks lookup returns the longest cached prefix and a
  load protocol in which every rank issues the same collectives
  (\S\ref{sec:shard}).
\item A \textbf{rank-partitioned remote tier}: because the sharded store
  partitions chunks statically by rank, rank $r$ of every engine can use its own
  server instance without a cluster protocol or client-side hashing, while
  cross-engine reuse still works (\S\ref{sec:perrank}).
\item \textbf{Asynchronous loading for the sharded store} and fixes for
  \textbf{three client bugs}, each with a regression test (\S\ref{sec:async},
  \S\ref{sec:bugs}).
\item An \textbf{evaluation} by trace replay of production traffic: p99 on a
  hit-heavy stream falls from \NHOnebpNineNine\s to \NHFivebpNineNine\s, a
  $124\times$ reduction, with 98.5\% of prompt tokens served from L2; the
  sharded store raises prefill throughput by 14--34\%; and a five-node
  disaggregated deployment using these changes passes the workload's SLA
  (\S\ref{sec:eval}).
\end{itemize}

The broader lesson is that on current GPUs recompute is a demanding baseline.
A cache tier has to be judged against the prefill rate of the hardware it serves,
not against the latency of the network.

\section{Background}
\label{sec:background}

\paragraph{MLA and its KV cache.} MLA~\cite{deepseekv2} projects each token's keys
and values into one low-rank latent vector shared by all attention heads. The
per-token KV state is therefore small and independent of the head count. For
GLM-5.3 in FP8, one 256-token LMCache chunk covering all layers is 12.4\,MB, so
$s\approx\NsKB$\,KB per token. Splitting attention heads across GPUs (tensor
parallelism) does not shrink this state, because every head needs the full latent.

\paragraph{Prefill context parallelism.} PCP~\cite{ringattention,ulysses,cpinference}
splits the prompt along the sequence, so each of the $W$ ranks processes about
$1/W$ of the tokens. Ranks exchange KV with an AllGather so that causal attention
sees every earlier token; after the exchange every rank holds the whole prompt's
latent KV. In our deployment the MoE layers additionally use expert parallelism
(EP8). Every forward step contains NCCL collectives, so \emph{the ranks run in
lockstep}: anything that delays one rank delays all of them.

\paragraph{Prefill/decode disaggregation.} Prefill is compute-bound and decode is
memory-bandwidth-bound, so separating them lets each use its own parallelism and
batch limits~\cite{distserve,splitwise}. A router sends each request to a prefill
engine, which computes (or loads) the prompt's KV and ships it to a decode engine
over NIXL~\cite{nixl}; the decode engine streams the answer.

\paragraph{LMCache.} LMCache~\cite{lmcache} splits KV into fixed-size chunks keyed
by a hash of the token prefix, so a chunk's key identifies everything before it.
A \emph{lookup} returns the longest prompt prefix whose chunks are all present; a
gap means everything after it is recomputed. For MLA, LMCache stores and loads
only on worker 0 (\texttt{save\_only\_first\_rank}), since every tensor-parallel
rank holds the same latent KV. The remote backend speaks the RESP protocol to a
Redis-compatible server, through a native client that keeps a pool of
connections per engine worker.

\paragraph{Valkey.} Valkey~\cite{valkey}, like Redis, executes every command on
a single main thread; I/O threads only parallelize socket reads and writes. A
single instance's command throughput is therefore bounded by one core, whatever
the client count.

\begin{figure*}[t]
\centering
\begin{tikzpicture}[
  font=\footnotesize,
  box/.style={draw, rounded corners=2pt, align=center, inner sep=4pt, fill=white},
  rank/.style={draw, minimum width=0.78cm, minimum height=0.55cm, inner sep=0pt, fill=blue!10, align=center, font=\scriptsize},
  l1/.style={draw, minimum width=0.78cm, minimum height=0.32cm, inner sep=0pt, fill=orange!15, font=\scriptsize},
  vk/.style={shape=cylinder, draw, shape border rotate=90, aspect=0.22, minimum width=0.78cm, minimum height=0.72cm, inner sep=0pt, fill=green!12, font=\scriptsize},
  arr/.style={-{Latex[length=1.8mm]}, thick},
  darr/.style={{Latex[length=1.4mm]}-{Latex[length=1.4mm]}, semithick}]
  \foreach \i in {0,...,7} {
    \node[rank] (r\i) at (\i*1.0, 0) {rank \i};
    \node[l1, below=0.08cm of r\i] (l\i) {L1};
    \node[vk] (v\i) at (\i*1.0, -2.1) {$v_{\i}$};
    \draw[darr] (l\i.south) -- (v\i.north);
  }
  \node[draw, rounded corners=3pt, dashed, inner sep=6pt, fit=(r0)(r7)(l0)(l7),
        label={[anchor=south]north:Prefill engine on one 8$\times$B300 node: PCP8 + EP8, ranks in lockstep}] (pf) {};
  \node[box, left=1.5cm of pf] (router) {Router\\[1pt]{\scriptsize cache- and load-aware}};
  \node[box, right=1.7cm of pf, minimum height=1.1cm] (dec) {Decode node\\[1pt]{\scriptsize DP8 + EP8}};
  \draw[arr] (router) -- node[above]{\scriptsize prompt} (pf);
  \draw[arr] (pf) -- node[above]{\scriptsize KV} node[below]{\scriptsize NIXL} (dec);
  \node[draw, rounded corners=3pt, dashed, inner sep=5pt, fit=(v0)(v7),
        label={[anchor=north]south:L2: one Valkey instance per rank; $v_r$ is shared by rank $r$ of every prefill engine}] (l2) {};
  \node[anchor=north west, align=left, font=\scriptsize, text width=2.9cm] at ($(l2.north east)+(0.35,0.05)$)
    {Chunk $c$ is owned by rank $c \bmod 8$: only that rank stores, looks up and loads it.};
\end{tikzpicture}
\caption{System overview, final design. The original design stored every chunk
from rank 0 into a single Valkey instance.}
\label{fig:arch}
\end{figure*}
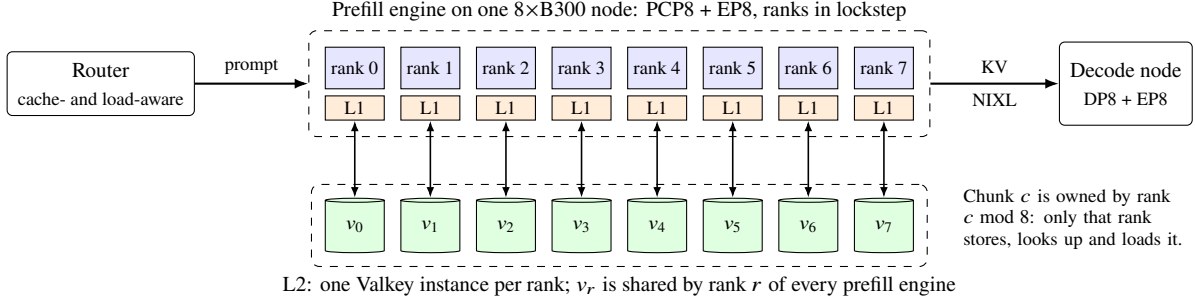

\section{Methodology}
\label{sec:method}

We wanted measurements that match production without touching it.

\paragraph{Trace replay.} From router logs and the benchmark client's database we
rebuilt the exact request stream that one production prefill engine received in
the 16k-token cell of the SLA workload: 683 requests with their original prompts
and send times. A replayer sends each request at its original offset and records
completion time, prompt tokens and cached tokens. Unless stated otherwise,
\emph{latency} is the prefill leg's completion time, which the router logs as
``prefill completed'' in production.

\paragraph{Isolated engine.} Replays ran on a free 8$\times$B300 node with the
production container image, command line and environment (vLLM~0.31.0,
LMCache~0.5.5 with the sharded-store patch, FlashInfer~0.7.0, NCCL~2.30, CUDA~13.0).
The differences from production were: prefill only (no decode handoff), the
LMCache controller off (so no production router could select the engine), and a
private Valkey~8.1 (400\,GB, LRU, 16 I/O threads) instead of the shared 1.28\,TB
instance. A sampler recorded engine metrics every 5\s and Valkey
\texttt{INFO commandstats} for server-side GET counts and time.

\paragraph{Hit-heavy subset.} The full stream's chunks ($\approx$495\,GB) exceed
the private Valkey, and replaying it in order under LRU evicts each prompt just
before it is needed again, so hit rates stay at 10--20\%. To measure the hit path
we use the first 450 requests (\NSUBtoktotal M prompt tokens, median
\NSUBtokpFiveZero\ tokens, range \NSUBtokmin--\NSUBtokmax), which fit. Before each
arm, Valkey is flushed, primed by one replay, and the engine restarted so that L1
is empty; 98.5\% of prompt tokens are then served from L2. The same 450 requests
in the cold run give the recompute baseline.

\paragraph{Profiling.} Nsight Systems captured all eight GPU workers with CUDA and
NVTX tracing, started and stopped through vLLM's profiler hook; py-spy took stack
dumps of engine processes during stalls with no GPU activity. Profiled runs are
used only to attribute time, never as latency measurements.

\paragraph{Reporting.} Percentiles use linear interpolation over all requests of
a run. Every number in \S\ref{sec:eval} is computed by one script from the raw
per-request logs. Each condition was run once; \S\ref{sec:limits} discusses what
this does and does not support.

\section{A Cost Model for Remote KV Hits}
\label{sec:model}

Let $n$ be the number of prompt tokens served from the cache, $s$ the KV bytes per
token, $R$ the engine's prefill rate in tokens per second, and $B$ the bandwidth
the cache tier delivers to one prompt. To first order,
\begin{equation}
T_{\text{recompute}} = \frac{n}{R}, \qquad
T_{\text{load}} = T_0 + \frac{n s}{B},
\label{eq:model}
\end{equation}
where $T_0$ is everything a hit still pays that does not scale with
bandwidth: scheduling, lookup round-trips, the forward step over the few
uncached tail tokens, and the per-chunk metadata exchange and broadcast.

\paragraph{Break-even bandwidth.} Ignoring $T_0$, a hit pays off iff
$B > B^* = sR$. The prompt length cancels: the threshold depends only on the
model's KV footprint and the hardware's prefill speed. For GLM-5.3 FP8 on one
B300 node, $s=\NsKB$\,KB/token and $R=\NRrate$k tokens/s (the median rate of cold
misses in our replay), so
\[
B^* = sR \approx \NBstar~\text{GB/s}.
\]
A single Valkey instance serving eight PCP ranks delivered 1.71\GBs at best
(\S\ref{sec:eval-bw}), below $B^*$: every L2 hit was slower than recomputing.

\paragraph{Fixed cost.} With $T_0>0$, a hit pays off only if
$n/R > T_0 + ns/B$, which requires at least $n > R\,T_0$. Below that prompt length,
recomputing is faster at \emph{any} bandwidth. We fit $T_0$ on a single point, the
median of the fastest configuration (8 instances, $B=10.76$\GBs, \S\ref{sec:eval-hit}):
$T_0 = 0.40 - 0.79/10.76 \approx \NTzero$\s. The model then predicts
$T_0 + 0.79/1.71 = \Npredone$\s for a 16k-token hit on a single instance, close
to the \NBFivehitFiveZero--\NBThreehitFiveZero\s we measured in three runs not used
for the fit (\cref{fig:cost}). The implied threshold is $R\,T_0\approx\NNbreak$k
tokens: on this hardware and implementation, L2 cannot improve the median latency
of a 16k-token prompt, and pays off in latency only for longer prompts.
Its value at 16k is capacity: a hit consumes no GPU compute, which matters when
prefill engines are saturated.

\paragraph{Stability of synchronous loading.} With synchronous loading, a load
runs inside the engine step and, under PCP, stalls all ranks. The engine then
behaves like a single-server queue~\cite{kleinrock} whose service time per hit
is at least $ns/B$. With arrival rate $\lambda$, its utilization is
$\rho=\lambda\,ns/B$. In the hit-heavy replay $\lambda=\Nlam$ requests/s and a
median prompt is 0.79\,GB. Under that load the single instance sustained only
0.78--0.89\GBs (server-side measurement), so $\rho\approx\Nrhosust>1$: the queue
has no steady state and grows for as long as hits keep arriving. With eight
instances, $\rho\approx\Nrhoeight$.

\begin{figure}[t]
\centering
\includegraphics[width=\columnwidth]{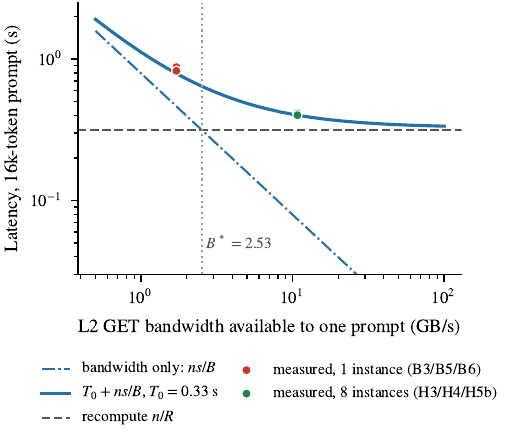}
\caption{Model (\cref{eq:model}) against measurement for a 16k-token hit.
$T_0$ is fitted on the H5b median only; the single-instance runs were not used to
fit. The dashed line is the measured recompute median.}
\label{fig:cost}
\end{figure}

\section{Design}
\label{sec:design}

The model points to two levers: deliver more bandwidth per prompt, and keep the
engine running while loads are slow. Both start from where the KV lives under PCP.

\subsection{Sharding the MLA KV Store Across PCP Ranks}
\label{sec:shard}

\paragraph{The rank-0 bottleneck.} Under tensor parallelism, storing MLA KV only
from rank 0 is sound: every rank holds the same latent. Under PCP the latent is
also replicated after the AllGather, but now rank 0 copies the entire prompt's KV
to host memory \emph{inside} the step, while the other seven ranks wait for it at
the next collective. In an Nsight trace of two concurrent $\sim$200k-token prompts,
rank 0 wrote 19.1\,GB to host memory and was 82\% busy against 98\% on the other
ranks, which accumulated $\sim$1.55\s of AllGather wait each in a 9.5\s window.

\paragraph{Ownership.} Since every rank holds the whole prompt's KV after the
AllGather, any rank can store any chunk. Chunk $c$, the chunk that starts at token
$256c$, is owned by rank $o(c)=c \bmod W$. Each rank stores, looks up and loads only
its own chunks, and its L1 is a $1/W$ slice of the engine's host-memory budget, so
total memory is unchanged.

\paragraph{Lookup.} The vLLM scheduler already sends each lookup to every rank's
lookup server and takes the minimum of the answers. In shard mode, rank $r$
answers with the start of its first owned chunk that is missing,
$a_r=\min\{c : o(c)=r,\ c \text{ missing}\}$, or $N$ if it has all of them.

\begin{proposition}
Let a prompt have chunks $0,\ldots,N-1$ and let $m$ be the first chunk missing
on its owner ($m=N$ if none). Then $\min_r a_r = m$.
\end{proposition}
\begin{proof}
Every chunk $c<m$ is present, so no rank has a missing owned chunk below $m$, and
$a_r\ge m$ for all $r$. If $m<N$, chunk $m$ is missing and owned by $o(m)$, and
all of that rank's owned chunks below $m$ are present, so $a_{o(m)}=m$.
\end{proof}

The minimum the scheduler already computes is therefore exactly the longest prefix
whose chunks are all present: sharded lookup needs no new messages. We also checked
the claim exhaustively for up to 10 chunks and $W\in\{2,3,4,8\}$, and on random
prompts of up to 800 chunks.

\paragraph{Load.} Loading must never let one rank skip a collective that the others
enter, or the step hangs. Every rank therefore runs the same sequence regardless of
what it found locally (\cref{alg:load}): fetch and stage its owned chunks, then
exchange one small message per rank containing a fingerprint of the chunk list, its
first failure and its chunk metadata; compute the agreed prefix from the $W$
messages, which are identical on every rank; and broadcast each chunk of the prefix
from its owner. Local failures (eviction, a failed GET, an unhealthy engine, a
staging error) are reported only through the exchange and become an agreed, shorter
prefix; mismatched inputs produce different fingerprints and an agreed prefix of
zero. Multi-process tests with real cache engines check these invariants under
random failures and after repeated failed loads.

\begin{algorithm}[t]
\caption{Sharded load on rank $r$ (all ranks run it)}
\label{alg:load}
\begin{algorithmic}[1]
\State $C \gets$ chunk list of (tokens, mask);\; $F \gets$ fingerprint$(C)$
\State $H, f_r \gets$ fetch owned chunks of $C$; $f_r$ = first owned failure
\State stage $H$ on the local GPU (non-blocking)
\For{$q = 0,\ldots,W-1$} \Comment{same order on every rank}
  \State $M_q \gets$ \textsc{BroadcastObject}$((F, f_r, \mathrm{meta}(H)),\ \mathrm{src}=q)$
\EndFor
\State $P \gets$ \textsc{AgreePrefix}$(M_0,\ldots,M_{W-1})$ \Comment{0 if fingerprints differ}
\For{$c = 0,\ldots,P-1$}
  \State \textsc{Broadcast}$(\mathrm{chunk}_c,\ \mathrm{src}=c \bmod W)$
\EndFor
\State write chunks $<P$ to the KV cache; release owned buffers
\end{algorithmic}
\end{algorithm}

\paragraph{Remote keys.} For MLA, LMCache names every remote key after worker 0 so
that ranks share one copy. In shard mode every rank writes its own chunks under
that worker-0 name, which keeps a chunk's remote key independent of the rank that
wrote it and preserves reuse across engines.

\subsection{A Rank-Partitioned Remote Tier}
\label{sec:perrank}

Raising $B$ needs more server threads, hence more Valkey instances. Spreading keys
over instances usually requires a cluster protocol or client-side consistent
hashing. The sharded store already provides a static partition: rank $r$ only ever
touches chunks with $c \bmod W = r$, and rank $r$ of \emph{every} engine owns the
same chunk indexes. Instance $r$ can therefore hold exactly rank $r$'s chunks, with
no routing logic, no traffic between instances, and no loss of cross-engine reuse.
The client change is one substitution: \texttt{\{rank\}} in the remote URL becomes
the worker's rank, e.g.\ \texttt{resp://valkey-\{rank\}:6379}. Each rank's share of
a 16k-token prompt ($\sim$100\,MB) is then fetched from its own instance in
parallel with the other seven.

\subsection{Asynchronous Loading for the Sharded Store}
\label{sec:async}

With enough bandwidth, $\rho\ll1$ and synchronous loads keep up
(\S\ref{sec:model}). Asynchronous loading is still valuable because it removes the
failure mode: a slow load delays only its own request while the engine keeps
running others. LMCache's asynchronous path assumed rank 0 performs all lookups;
we extend it to shard mode without new protocol messages. Each rank prefetches only
its owned chunks and passes the storage layer a custom cumulative-length vector
whose $k$-th entry is the start token of its $k$-th owned chunk. The storage
layer's ``contiguous hits'' answer is then exactly $a_r$ from the proposition,
in tokens, and the scheduler's minimum is the same prefix as in the synchronous
path. A property test over 3{,}000 random presence patterns checks this
equivalence. One edge case changed: a rank with no hits used to report 0 tokens,
which is wrong under sharding (if rank 1's first owned chunk is missing, the
correct answer is 256 tokens, since chunk 0 belongs to rank 0). It now reports
the first entry of the vector, which is 0 outside shard mode.

\subsection{Three Client Bugs}
\label{sec:bugs}

Bandwidth explained the steady cost; three bugs explained the worst outliers, all
of which left an empty GPU timeline.

\paragraph{Lost completions.} The native client's synchronous methods submitted a
request and only then registered the future that would receive the result, while
completions were drained on an event-loop thread. A completion drained between
the two steps was dropped, and the caller blocked forever. Sharding multiplies the
exposure by $W$, since every lookup now goes to every rank, and one loss is
permanent: afterwards every scheduler lookup waits its full 3\s timeout. One
production engine logged exactly this, one timeout per routed request (16 of 16 in
one minute) from one moment on, and py-spy showed every lookup thread blocked in
\texttt{fut.result()} with nothing outstanding. The fix holds one lock across
submit-and-register and across the drain. A test with a fake client that drains
before \texttt{submit} returns fails without the fix; replays showed no measurable
latency cost.

\paragraph{A GET of a missing key never returns.} RESP encodes a missing value as
the 5-byte null bulk string \texttt{\$-1\textbackslash r\textbackslash n}. The
native connector read a fixed 11-byte header sized for a chunk, so on a missing key
the read never completed and the connection was lost for good. Under LRU eviction
a chunk can disappear between the lookup (\texttt{EXISTS}) and the load
(\texttt{GET}); a synchronous load then stalls all ranks for the 10\s blocking
timeout. The fix parses the header up to CRLF, treats \texttt{\$-1} as a per-key
miss with the reply fully consumed, and makes the Python layer free unfilled
buffers instead of returning them as KV. Against a real Valkey, a 4-chunk GET with
one key deleted now returns \texttt{[T,T,F,T]} in 69\ms\ and the connection serves
the next GET in 40\ms; before, it produced no completion within 12\s.

\paragraph{Async GETs used the wrong key.} The asynchronous GET did not rewrite keys
to the worker-0 name used by the asynchronous lookup, so ranks other than 0 found
their chunks and then requested keys that do not exist (which is how we found the
previous bug). Upstream never exercised this, because only rank 0 performed
asynchronous lookups for MLA.

\section{Evaluation}
\label{sec:eval}

\subsection{Where a 16k-Token Prompt's KV Comes From}
\label{sec:eval-cost}

\Cref{tab:source} replays the full 683-request stream under different cache
states. Loading from L1 (host memory) is slightly faster than recompute; loading
from a single L2 instance is $2.7\times$ slower, on an idle engine (at most one
request running, none waiting). This matches the model: at 1.71\GBs a
0.79\,GB prompt needs 0.46\s of transfer on top of $T_0$.

\begin{table}[t]
\centering
\small
\caption{Prefill latency of the 683-request 16k-token stream by cache state.
``Hits'' are requests with $\ge$12k cached tokens. Single L2 instance, synchronous
loading unless noted.}
\label{tab:source}
\begin{tabular}{@{}llrrr@{}}
\toprule
Run & State & p50 & p99 & hits: n, p50 \\
\midrule
B1 & cold (recompute) & \NBOnepFiveZero & \NBOnepNineNine & \NBOnenhit, \NBOnehitFiveZero \\
B2 & L1 warm & \NBTwopFiveZero & \NBTwopNineNine & \NBTwonhit, \textbf{\NBTwohitFiveZero} \\
B3 & L1 cold, L2 warm & \NBThreepFiveZero & \NBThreepNineNine & \NBThreenhit, \textbf{\NBThreehitFiveZero} \\
B5 & as B3 (control) & \NBFivepFiveZero & \NBFivepNineNine & \NBFivenhit, \NBFivehitFiveZero \\
B6 & as B3, race fix & \NBSixpFiveZero & \NBSixpNineNine & \NBSixnhit, \NBSixhitFiveZero \\
A1 & as B3, async, rank-0 & \NAOnepFiveZero & \NAOnepNineNine & \NAOnenhit, \NAOnehitFiveZero \\
\bottomrule
\end{tabular}
\end{table}

\subsection{Remote-Tier Bandwidth}
\label{sec:eval-bw}

\Cref{fig:bw} measures aggregate GET bandwidth for 12.4\,MB chunks with LMCache's
native client against Valkey, with each client fetching 8-chunk batches (99\,MB,
one rank's share of a 16k prompt). Going from one client to eight on one instance
adds only 13\%, because the clients queue on the same server thread, and one
rank's batch takes 440\ms\ at p50 instead of 65\ms. With one instance per rank,
bandwidth scales to 10.76\GBs ($6.3\times$) and a batch takes 70\ms. During the
hit-heavy runs on a single instance, Valkey's main thread spent 73--74\% of
wall-clock time executing GETs (89--91\% in the busiest minute).

\begin{figure}[t]
\centering
\includegraphics[width=\columnwidth]{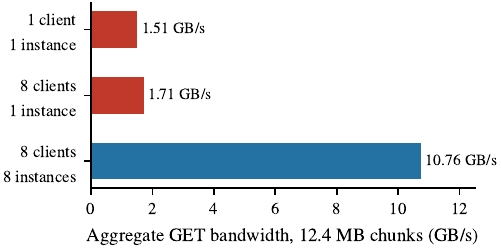}
\caption{Aggregate GET bandwidth, LMCache native RESP client, 12.4\,MB chunks.}
\label{fig:bw}
\end{figure}

\subsection{Hit-Heavy Traffic}
\label{sec:eval-hit}

\Cref{tab:hit,fig:trace,fig:cdf} compare five configurations on the 450-request
hit-heavy subset with 98.5\% of tokens served from L2 (except H2b; see below).

\begin{table*}[t]
\centering
\small
\caption{Hit-heavy subset (450 requests, 16k tokens, L2 primed, L1 empty).
Latencies in seconds. All runs completed 450/450 requests without errors.}
\label{tab:hit}
\begin{tabular}{@{}lllrrrrcrr@{}}
\toprule
Run & L2 & Loading & p50 & p90 & p99 & max & p50 by third of run & from L2 & GPU util \\
\midrule
H1b & 1 instance & sync & \NHOnebpFiveZero & \NHOnebpNineZero & \NHOnebpNineNine & \NHOnebmax & \NHOnebthirdOne\ $\to$ \NHOnebthirdTwo\ $\to$ \NHOnebthirdThree & \NHOnebcached\% & 7\% \\
H2b & 1 instance & async (rank-0 store) & \NHTwobpFiveZero & \NHTwobpNineZero & \NHTwobpNineNine & \NHTwobmax & \NHTwobthirdOne\ $\to$ \NHTwobthirdTwo\ $\to$ \NHTwobthirdThree & \NHTwobcached\% & 28\% \\
H3 & 8 instances & sync, sharded & \NHThreepFiveZero & \NHThreepNineZero & \NHThreepNineNine & \NHThreemax & \NHThreethirdOne\ $\to$ \NHThreethirdTwo\ $\to$ \NHThreethirdThree & \NHThreecached\% & 21\% \\
H4 & 8 instances & async, sharded & \NHFourpFiveZero & \NHFourpNineZero & \NHFourpNineNine & \NHFourmax & \NHFourthirdOne\ $\to$ \NHFourthirdTwo\ $\to$ \NHFourthirdThree & \NHFourcached\% & 20\% \\
H5b & 8 instances & async, sharded, all fixes & \textbf{\NHFivebpFiveZero} & \textbf{\NHFivebpNineZero} & \textbf{\NHFivebpNineNine} & \textbf{\NHFivebmax} & \NHFivebthirdOne\ $\to$ \NHFivebthirdTwo\ $\to$ \NHFivebthirdThree & \NHFivebcached\% & 22\% \\
\midrule
\multicolumn{3}{@{}l}{Recompute, same requests (B1)} & \NRECsubpFiveZero & \NRECsubpNineZero & \NRECsubpNineNine & \NRECsubmax & & 0.8\% & \\
\multicolumn{3}{@{}l}{Production engine, same requests} & \NPRODsubpFiveZero & \NPRODsubpNineZero & \NPRODsubpNineNine & \NPRODsubmax & & & \\
\bottomrule
\end{tabular}
\end{table*}

\begin{figure}[t]
\centering
\includegraphics[width=\columnwidth]{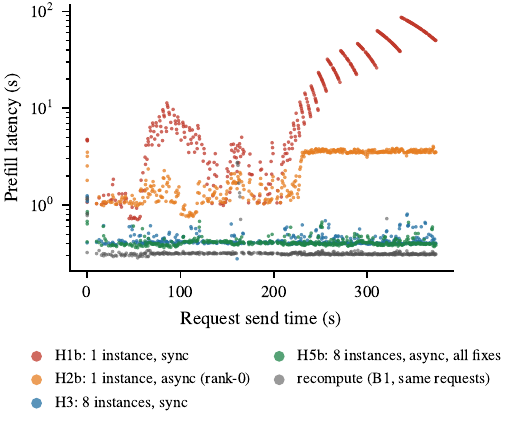}
\caption{Per-request prefill latency over the hit-heavy replay (log scale). With one
instance and synchronous loading (H1b) the engine never catches up; with one instance
per rank (H3, H5b) latency is flat from the first request to the last.}
\label{fig:trace}
\end{figure}

\paragraph{One instance, synchronous (H1b).} The run behaves as the model predicts
for $\rho>1$: median latency grows from \NHOnebthirdOne\s in the first third of
the run to \NHOnebthirdThree\s in the last, up to \NHOnebmaxrun\ requests are in
the engine at once (\cref{fig:backlog}), and the GPUs are 7\% busy. Every request
completes, but the engine spends most of its time waiting for Valkey.

\paragraph{One instance, asynchronous (H2b).} Asynchronous loading with the
default rank-0 store bounds latency at 3.2--3.8\s, but only by giving up: 49\% of
requests hit the 3\s asynchronous-lookup timeout and recomputed while their
prefetch continued in the background, which is why only \NHTwobcached\% of tokens
came from L2 despite the same number of GETs as H1b.

\paragraph{Eight instances (H3--H5b).} With one instance per rank, median latency
is \NHThreepFiveZero\s from start to finish, p99 drops to \NHThreepNineNine\s, and
the engine never holds more than one request. Asynchronous loading (H4) matches the
synchronous median, as expected at $\rho\approx\Nrhoeight$, and trims the tail.
With all fixes (H5b), p99 is \NHFivebpNineNine\s, $124\times$ below H1b, and the
latency distribution sits close to the production engine's on the same requests,
which served only 1--35\% of tokens from any cache (\cref{fig:cdf}). The Valkey
shards averaged 8.9\ms\ of server time per 12.4\,MB GET.

\paragraph{Remaining gap to recompute.} At the median, H5b is 0.09\s slower than
recomputing the same prompts. Transfer accounts for only \NTloadeight\s of a hit
at 10.76\GBs; the rest is $T_0$, spent in the engine (scheduling, the forward step
over the uncached tail, and the per-chunk exchange and broadcast) rather than in
Valkey. Even at infinite bandwidth a 16k-token hit would take about \NTzero\s,
slightly more than recompute, so on a fast GPU it cannot win on median latency
until $T_0$ shrinks (\S\ref{sec:discussion}). The tail is better than recompute's (p99
\NHFivebpNineNine\s against \NRECsubpNineNine\s), and a hit uses no GPU compute.

\begin{figure}[t]
\centering
\includegraphics[width=\columnwidth]{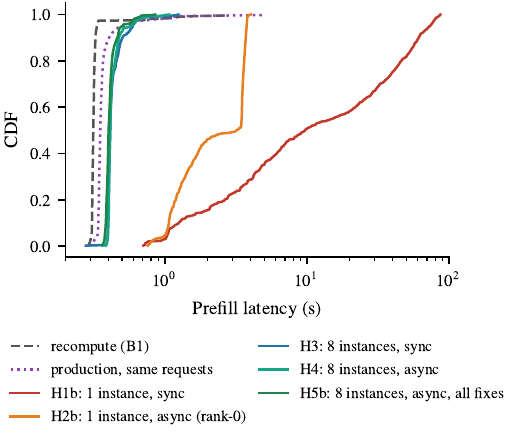}
\caption{Latency CDFs on the hit-heavy subset (log scale), with the recompute
baseline and the production engine on the same requests.}
\label{fig:cdf}
\end{figure}

\begin{figure}[t]
\centering
\includegraphics[width=\columnwidth]{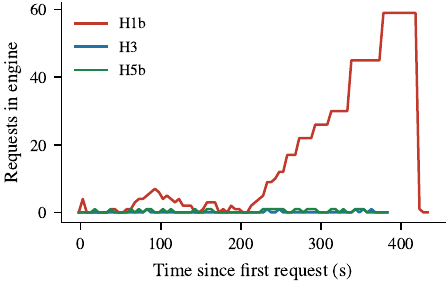}
\caption{Requests running or waiting in the engine (5\s samples).}
\label{fig:backlog}
\end{figure}

\subsection{Prefill Throughput}
\label{sec:eval-tput}

\Cref{fig:tput} measures computed prompt tokens per second on one node with 8
concurrent $\sim$200k-token prompts for 3 rounds, with LMCache storing every
chunk. The sharded store alone raises throughput from 43.2k to 49.4k tokens/s
(+14\%). The same traces exposed a second, unrelated bottleneck: FlashInfer's
fused SwiGLU + FP8 re-quantization kernel took 1.57\ms\ per call and 17\% of
prefill GPU time, because under EP8 most of its thread blocks process slots routed
to other GPUs. An upstream kernel rewrite~\cite{fi3969}, which we built for SM103
and validated, reduces it to 0.23\ms. Each fix alone gives +14--15\%; together
they give 57.3--58.0k tokens/s (+34\%), within 4--5\% of running with LMCache
turned off. In the two-prompt Nsight window, total AllGather time fell from 15.8\s
to 5.8\s and rank 0's device-to-host traffic from 19.1\,GB to 2.4\,GB.
With a remote backend attached, puts equal stored chunks and throughput is
unchanged (57.3k against 57.7k tokens/s).

\begin{figure}[t]
\centering
\includegraphics[width=\columnwidth]{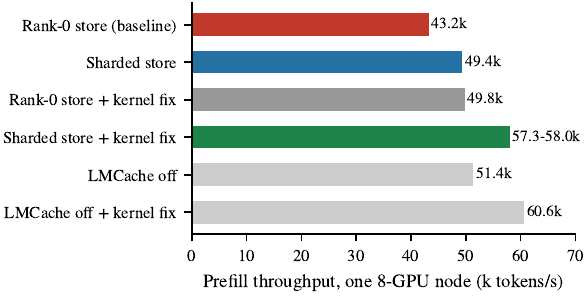}
\caption{Prefill throughput on one 8$\times$B300 node, 8 concurrent
$\sim$200k-token prompts. ``Kernel fix'' is FlashInfer PR \#3969~\cite{fi3969}.}
\label{fig:tput}
\end{figure}

\subsection{Correctness}
\label{sec:eval-correct}

Bit-exact comparison is not usable: the engine is not deterministic, and two cold
runs of the same prompt already produce different text. We test semantics with six
needle-in-a-haystack prompts of 12.8k--166.7k tokens, each hiding a random code at a
random depth, answered with greedy decoding four ways: cold; from L1 via
asynchronous prefetch; from L2 via asynchronous prefetch after an engine restart
(12{,}800--166{,}656 tokens loaded); and cold again as a control. The code was
found in 24 of 24 answers, and cache-served answers differ from cold ones no more
than the two cold runs differ from each other. All runs logged zero lookup
timeouts, zero asynchronous timeouts and zero errors. The synchronous sharded
store passed 3/3 needles after GPU-cache eviction (174k--204k tokens loaded),
36/36 answers in a concurrent burst, and a 30-minute soak of 1{,}549 requests
without errors.

\subsection{End-to-End Deployment}
\label{sec:eval-e2e}

Finally, we ran the full configuration (sharded store, eight L2 instances,
asynchronous loading, all fixes) as a disaggregated deployment with four PCP8
prefill nodes and one decode node (DP8, EP8) on the workload's 200k-token SLA
benchmark, which sends a sustained closed-loop mix of long prompts and checks TTFT
against P90 $<$ 70\s, P50 $<$ 30\s and mean $<$ 30\s. \Cref{tab:e2e} and
\cref{fig:e2e} compare it with earlier deployments that passed the same benchmark.
This comparison changes more than the cache tier: the aggregated baseline also
differs in parallelism and in prefill/decode placement. We report it as deployment
context, not as an ablation of the techniques above.

\begin{table}[t]
\centering
\small
\caption{200k-token SLA benchmark. TPM: input plus output tokens per minute.}
\label{tab:e2e}
\setlength{\tabcolsep}{4pt}
\begin{tabular}{@{}lrrrr@{}}
\toprule
 & Aggr. & 6P+1D & 5P+1D & \textbf{4P+1D} \\
\midrule
B300 nodes & 11 & 7 & 6 & \textbf{5} \\
Success & 100\% & 99.9\% & 99.88\% & \textbf{100\%} \\
TTFT P90 (s) & 34.0 & 25.3 & 44.5 & \textbf{30.5} \\
TTFT P50 (s) & 13.1 & 15.1 & 28.4 & \textbf{18.0} \\
TTFT mean (s) & 19.9 & 17.0 & 29.6 & \textbf{20.1} \\
TPM (M) & 7.74 & 14.1 & 12.6 & \textbf{12.7} \\
TPM per node (M) & 0.70 & 2.01 & 2.10 & \textbf{2.55} \\
\bottomrule
\end{tabular}
\end{table}

\begin{figure}[t]
\centering
\includegraphics[width=\columnwidth]{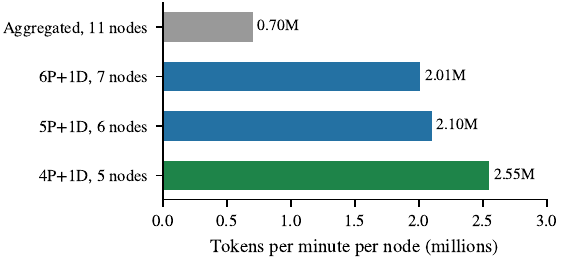}
\caption{Tokens per minute per B300 node on the SLA benchmark.}
\label{fig:e2e}
\end{figure}

\section{Discussion}
\label{sec:discussion}

\paragraph{Judge a cache tier against the prefill rate.} $B^*=sR$ grows with
accelerator speed: every faster GPU generation raises the bandwidth a cache tier
must deliver to be worth reading. At \NBstar\GBs per prompt for one B300 node,
a cache tier built from general-purpose key-value servers needs to scale out with
the engine's parallelism, as the rank partitioning does here, or move to RDMA
transports~\cite{mooncake}. Compression~\cite{cachegen} lowers $s$, and loading
and computing different parts of a prompt concurrently~\cite{cake} avoids treating
the choice as all-or-nothing.

\paragraph{Shrink the fixed cost.} Our $T_0\approx\NTzero$\s means that, on this
implementation, remote hits on prompts below $\sim$\NNbreak k tokens are slower
than recompute at any bandwidth. The load path broadcasts each chunk from its
owner, one collective per chunk; for a 16k prompt that is 64 broadcasts plus $W$
metadata exchanges. Batching chunks per owner into a single AllGather-style
collective, and overlapping the exchange with the fetch, are natural next steps.
A load-or-recompute policy that compares the model's predicted load time with
$n/R$ per request would also stop short prompts from paying $T_0$.

\paragraph{Static partitions make scale-out free.} The rank partitioning works
because the sharded store fixes ownership by chunk index, a property any
deterministic, rank-local placement shares. Scaling to more instances than ranks
(for capacity rather than bandwidth) needs a second-level hash within each rank's
partition; scaling the instance count with $W$ keeps the property that one rank
talks to one server.

\paragraph{Lockstep amplifies everything.} Under PCP, any per-rank stall (a rank
doing all the stores, a slow GET, a lost completion) stalls all ranks. Designs for
context-parallel engines should keep per-rank work symmetric and keep any
operation that can block for an unbounded time out of the step.

\section{Limitations}
\label{sec:limits}

\begin{itemize}
\item \textbf{Single runs.} Each condition was run once. The large effects
  ($124\times$ at p99, $6.3\times$ bandwidth, the unstable queue) are far outside
  run-to-run variation, but differences under about 10\% (for example H3 against
  H4) are not established.
\item \textbf{Prefill only, one engine.} The replays omit the KV handoff to decode
  and use one engine as the only Valkey client; contention from several engines
  on shared instances is estimated, not measured. At the workload's
  $\sim$0.45 requests/s per engine and 0.79\,GB per prompt, six engines would
  draw about 2\GBs, well below 10.76\GBs.
\item \textbf{One prompt-length bucket.} Latency results are for 14.7k--17.9k-token
  prompts. $T_0$ is fitted on one point, and the prompt-length threshold it implies
  is a prediction that we have not yet measured directly.
\item \textbf{Cache state.} The hit-heavy subset was chosen to fit in the private
  instance; production's full, evicting instance would produce more misses and
  more \texttt{EXISTS}-to-\texttt{GET} races.
\item \textbf{End-to-end comparison.} \S\ref{sec:eval-e2e} compares deployments
  that differ in more than the cache tier.
\end{itemize}

\section{Related Work}
\label{sec:related}

\paragraph{KV-cache reuse and tiering.} vLLM's PagedAttention~\cite{vllm} and
SGLang's RadixAttention~\cite{sglang} reuse prefixes within an engine.
CachedAttention~\cite{cachedattention} keeps conversation KV in host memory and
disk and overlaps loading with computation; Mooncake~\cite{mooncake} builds a
disaggregated KV store over RDMA for a production chatbot; MemServe~\cite{memserve}
adds an elastic memory pool for disaggregated serving; LMCache~\cite{lmcache}
provides the multi-tier cache layer we build on. These systems motivate remote
tiers; we characterize when such a tier is slower than recompute on fast GPUs and
show how context parallelism enables a coordination-free scale-out of the tier.

\paragraph{Load versus recompute.} CacheGen~\cite{cachegen} compresses KV to cut
transfer time over networks. Cake~\cite{cake} computes and loads different parts
of the prompt concurrently, and HCache~\cite{hcache} restores state from cached
hidden states, trading transfer for partial recomputation. Our cost model makes the
underlying trade-off explicit for a remote tier, including a fixed cost that
bandwidth cannot remove; these techniques act on its $s$ and $B$ terms.

\paragraph{Disaggregation and context parallelism.} DistServe~\cite{distserve} and
Splitwise~\cite{splitwise} separate prefill from decode; DéjàVu~\cite{dejavu}
streams KV between stages. Ring Attention~\cite{ringattention},
DeepSpeed-Ulysses~\cite{ulysses} and context-parallel inference~\cite{cpinference}
split long sequences across devices. We show that combining MLA, context
parallelism and a KV cache layer designed for tensor parallelism concentrates
cache I/O on one rank, and give a sharded store whose correctness follows from the
existing lookup protocol.

\section{Conclusion}

A remote KV tier is only useful if it beats recomputation, and on current GPUs that
is a high bar: for GLM-5.3 on a B300 node, a cache hit must arrive at more than
\NBstar\GBs, and a fixed engine cost puts short prompts out of reach altogether.
A single key-value server could not meet the bar, and synchronous loading turned
the shortfall into an unstable queue. Sharding the MLA KV store across
context-parallel ranks spreads cache I/O evenly, and its static ownership lets the
remote tier scale out by rank with no coordination. Together with asynchronous
loading and three client fixes, this brings hit-heavy p99 from \NHOnebpNineNine\s
to \NHFivebpNineNine\s.

\section*{Availability}
The changes are submitted upstream to LMCache:
the sharded store~\cite{pr5545}, the per-rank remote URL~\cite{pr5548},
asynchronous loading for the sharded store~\cite{pr5553}, and the three client
fixes~\cite{pr5547,pr5550,pr5551}.

{\small
\bibliographystyle{plainnat}
\bibliography{refs}
}

\end{document}